\documentclass[conference]{IEEEtran}
\usepackage{bm}
\usepackage{graphicx}
\usepackage{latexsym}
\usepackage{amsmath}
\usepackage{amsthm} 
\newenvironment{aquote}[1]{%
  \pushQED{#1}%
  \begin{quote}
}{%
  \par\nointerlineskip\noindent\hfill(\popQED)%
  \end{quote}%
}

\usepackage{txfonts}
\usepackage[bookmarks=false]{hyperref}
\usepackage{stfloats}
\usepackage{algorithm}
\usepackage{algorithmic}

\begin{document}

\title{Online Assignment Algorithms for Dynamic Bipartite Graphs}

\author{\IEEEauthorblockN{Ankur Sahai}
\IEEEauthorblockA{Department of Computer Science,}
\IEEEauthorblockA{Indian Institute of Technology, Kanpur.}
Email: ankur.sahai@gmail.com}

\hyphenpenalty=100000
\setlength{\parindent}{10pt}
\setlength{\parskip}{1ex} 
\maketitle

\begin{abstract}
This paper analyzes the problem of assigning weights to edges incrementally in a dynamic complete bipartite graph consisting of producer and consumer nodes. The objective is to minimize the overall cost while satisfying certain constraints. The cost and constraints are functions of attributes of the edges, nodes and online service requests. Novelty of this work is that it models real-time distributed resource allocation using an approach to solve this theoretical problem.

This paper studies variants of this assignment problem where the edges, producers and consumers can disappear and reappear or their attributes can change over time. Primal-Dual algorithms are used for solving these problems and their competitive ratios are evaluated.
\end{abstract}

\section{Motivation}
\label{sec:motivation}
As more and more data moves to the cloud every day, it becomes important to analyze automated resource management schemes based on online algorithms that make the best use of the available storage while guaranteeing optimal performance to the users. This paper studies the theoretical aspects of the problem of allocating storage to VMs optimally.

The VMs running in a distributed system can be considered as the producers of I/O demands and data-centers as the consumers and this configuration can be visualized as a bipartite graph. This graph is complete as the I/O demand generated by any VM can be assigned to any of the data-centers and vice versa. The capacity of data-centers is equivalent to the capacity of the consumers. The average time per I/O operation or latency between a VM and a data-center can be compared with an edge distance. Failures of edges, consumers and producers can be likened to network link outage, failure of data-center / Storage Area Network (SAN) and VMs respectively and this contributes to the dynamic nature of the bipartite graph.

Allocation of storage demands generated by VMs to data-centers forms the assignment problem in the dynamic bipartite graph. As the I/O demands arrive incrementally and they need to be satisfied instantly, this forms the online part of the problem. Change in edge distances with time, can be visualized as being caused by a mobile user. The \emph{assignment without reallocation} constraint in problem \ref{sec:problem-2} that prevents reallocation of weights; simulates the practical limitation involved in moving large amounts of data across data-centers within a short period of time. It is to be noted that the available capacity of data-centers is a non-increasing function of time.

This paper presents three points of novelty:
\begin{enumerate}
\item Provides a unique theoretical perspective to resource allocation in distributed systems.
\item Analyzes the scenarios where properties of users, resources or the network link between them can change over time.
\item Extends this theoretical approach to mobile users.
\end{enumerate}

\section{Problem Definition}
\label{sec:problem-definition}

This paper aims to analyze online algorithms \ref{sec:online-algorithms} for dynamically evolving undirected graphs \cite{rand-dyn-graphs,emp-dyn-graph-algos,multiflow-dyn-graph-algos}. Given a complete bipartite graph $G = (V, E)$ where, ${V}$ is a finite set of nodes which consists of producers $i \in P$ and consumers $j \in C$ such that, ${V} = {P} \cup {C}$ and edges $e_{ij} \in E$ with distances $d_{ij}$, where $E = \{e_{ij} \; | \; i \in P, j \in C\}$.

A sequence of online \emph{service requests} $R = R(t), R(t+1), R(t+2), \cdots$ that are received as input specify either {\bf(a)} consumer demands {\bf(b)} failure / restoration of edges, producers and consumers {\bf(c)} changes in their attributes. This characterizes the dynamic nature of the bipartite graph. Consumer demands act by either changing the edge weights $w_{ij}(t)$ or removing an edge $e_{ij}(t) = 0$. For simplicity, this paper assumes that each producer generates atmost one demand $R_i$ throughout its lifetime and that a unique service request is generated at each time instances $t$. T is the set of all instances at which the online service requests are received $T = \cup_i t_i$.

Find an $\alpha$-competitive online algorithm \ref{sec:online-algorithms} for satisfying the service requests ${R}$ that minimizes the distance-weighted sum of edge weights:

\begin{equation}
\sum_{i \in P, j \in C} w_{ij}(t)  \cdot  d_{ij}(t)   \cdot  e_{ij}(t) \le \alpha  \cdot  OPT(t), \forall t \in T
\label{operation 5}
\end{equation} 
where, $\alpha$ is a constant and OPT(t) is the output of the optimal offline algorithm for the input received in the time interval $[0, t], t \in T$. Such that, 

\begin{equation}
\sum_{j \in C} w_{ij}(t) = R_i, \forall i \in P
\label{operation 6}
\end{equation}
Equation (\ref{operation 6}) guarantees that demands generated by producers are satisfied. This will be referred to as \emph{producer demand constraint}.

\begin{equation}
\sum_{i \in P} w_{ij}(t) \le M_{j}(t), \forall j \in C
\label{operation 7}
\end{equation}
Equation (\ref{operation 7}) ensures that consumer capacities at time $t \in T$ are not exceeded. This will be referred to as \emph{consumer capacity constraint}.

Dynamic nature of the edges is characterized by,
\begin{equation}
e_{ij}(t) = \left\{
\begin{array}{c l}
1 & \mbox{if edge exists, $i \in P$, $j \in C$ at time $t \in T$} \\
0 & \mbox{otherwise}
\end{array}
\right.
\label{operation 4}
\end{equation} 

This paper first analyzes the optimal offline algorithm and then focuses on online algorithms for solving this problem.

\subsection{Assignment without reallocation}
\label{sec:problem-2}
Consider a generalization of the problem in \ref{sec:problem-definition} where the weights assigned to edges cannot be reallocated (this is called \emph{assignment without reallocation} constraint). In this case the weights allocated to edges are a non-decreasing function of time $w_{ij}(t+1) \ge w_{ij}(t)$ except for edge failures when $w_{ij}(t+1)=0, \forall i \in P, j \in C$.

\subsection{Assignment with varying edge distances}
\label{sec:problem-3}
Consider a generalization of \emph{assignment without reallocation} in \ref{sec:problem-2} where edge distances can change over time $\exists (t,  \overline{t}) \in T, t \neq \overline{t} \; : \; d_{ij}(t) \neq d_{ij}(\overline{t})$.

\subsection{Assignment with node addition / failure or attribute changes}
\label{sec:problem-5}
Consider a generalization of \emph{assignment without reallocation} in \ref{sec:problem-2} with addition / failure of producers / consumers, $\exists (t, \overline{t}) \in T, t \neq \overline{t}, \; : C(t) \neq C(\overline{t}), \; P(t) \neq P(\overline{t})$. This paper assumes that when a consumer $j \in C$ fails the data stored on it is wiped off, $\sum_{i \in P}w_{ij} = 0$. This data is recreated by going through the demands generated earlier and is stored on a different consumer.

This section also considers a generalization of this problem where consumer capacities can change over time as specified by the service requests, $\exists (t, \overline{t}) \in T, \; t \neq \overline{t} \; :M_j(t) \neq M_j(\overline{t}), \; \forall j \in C$.

\subsection{Offline Assignment with multiple producer requests}
\label{sec:problem-6}
Consider a producer that generated multiple requests $R_{i1},R{i2},\cdots$.

\section{Related Work}

Papers that study other theoretical aspects of assigning resources to users incrementally are: online algorithms \ref{sec:online-algorithms} for the \emph{k}-server problems \cite{k-server-offline,k-server-lower-bounds,k-server-competitive,k-server-decomposition,k-server-survey,k-server-randomized}, min-flow \cite{min-flow-related}, online matching \cite{online-matching,online-weighted-bip-match}, dynamic assignment \cite{dynamic-assignment} bipartite network flow \cite{bip-network-flow} and assignment problem\cite{hungarian-algo, dynamic-hungarian-algo, incr-assign}. These are the most relevant results for online resource allocation problem \ref{sec:problem-definition}.

This problem \ref{sec:problem-definition} is also important because it can be used for distributed resource scheduling schemes such as the one used in VMware's virtualization framework \cite{vmware-drs, vmware-scale-storage}, Virtual Infrastructure using VirtualCenter,  a centralized distributed system and recently in VSphere, a cloud OS. The authors of VMware's Scalable Storage Performance white paper \cite{vmware-scale-storage} say:

\begin{aquote} {\cite{vmware-scale-storage}, page 2, para 1}
"Latency depends on many factors, including queue depth or capacity at various levels; I/O request size; disk properties such as rotational, seek, and access delays; SCSI reservations; and caching or prefetching algorithms."  
\end{aquote}

that the "latency" (time taken to complete I/O request which corresponds to cost of the objective function \ref{operation 5}) depends on "factors" that correspond to the attributes of the producers, consumers and network link considered by this paper \ref{sec:problem-definition}.

Distributed resource allocation \cite{job-assign-scalable,chromatic-sums,constraint-driven-scheduling}, fairness of resource allocation \cite{fairness-alloc,balanced-alloc} and dynamic load balancing \cite{related-load-balancing,online-load-balancing,hierearchical-load-balancing} issues have also been studied earlier. A survey of schemes for large scale cloud-computing platforms is presented in \cite{cloud-computing}.An analysis of the various distributed resource allocation techniques is presented in \cite{survey-distributed-computing}.

\section{Offline Algorithms}
The optimal offline algorithm has to exhaustively look at the available edges. This paper uses Linear Programming (LP) for solving the offline version due to ready availability of LP code \cite{lp-solve} that is used for verifying the output for different problem instances.

\subsection{LP}
\label{sec:offline-LP}
Linear Programming \cite{interior-point-methods,lp-onlline} is a method used to solve large-scale optimization problems with a set of constraints and an objective function (which has to be either minimized or maximized) both being linear.

The LP formulation for this problem is as follows,

Objective function:

Minimize:
\begin{equation}
\sum_{i \in P, j \in C} d_{ij}(t)  \cdot  w_{ij}(t), \; \; \; d_{ij}(t) \ge 0, \; w_{ij}(t) \ge 0
\label{operation 8}
\end{equation}

This LP is used for calculating the optimal solution OPT(t) for the input received in the time interval $[0, t], t \in T$. As the demands $R_i, i \in P$ are non-negative the weight assignments $w_{ij}$ are also non-negative. Edges that fail $e_{ij}(t) = 0$ (in \ref{operation 4}) have their corresponding $d_{ij}(t) = \infty$ so that, they are not selected.

Constraints:
\begin{equation}
\sum_{j \in C} w_{ij}(t) \ge R_i, \forall i \in P
\label{operation 9}
\end{equation}

Equation (\ref{operation 9}) represents the \emph{producer demand} constraint corresponding to (\ref{operation 6}).

\begin{equation}
\sum_{i \in P} w_{ij}(t) \le M_{j} \implies -\sum_{i \in P} w_{ij}(t) \ge - M_{j}, \; \forall j \in C
\label{operation 10}
\end{equation}

Equation (\ref{operation 10}) represents the \emph{consumer capacity} constraint corresponding to (\ref{operation 7}). In case the consumer capacity $M_j(t)$ changes at time $t \in T$ the corresponding \emph{consumer capacity} constraint is updated in the new LP formulation at time $t$

In addition to this, new constraints corresponding to the existing weight assignment on edges have to be added at each time instance $t \in T$.

\newtheorem{theorem}{Theorem}
\begin{theorem} {\bf (Correctness of LP \ref{sec:offline-LP})}
LP formulation in \ref{sec:offline-LP} produces a valid assignment of weights $w_{ij}$ on edges $e_{ij}$ corresponding to the demands R.
\end{theorem}
\begin{IEEEproof}
Equation (\ref{operation 9}) guarantees that the total demand generated by producers $i \in P$ is satisfied. Equation (\ref{operation 10}) ensures that the capacities of consumers $j \in C$ are not exceeded. By definition \ref{sec:problem-definition} this is a valid assignment of weight on edges.
\end{IEEEproof}

\begin{theorem} {\bf (Optimality of LP \ref{sec:offline-LP})}
LP formulation in \ref{sec:offline-LP} produces the optimal assignment of weights $w_{ij}$ on edges $e_{ij}$ corresponding to the demands in R.
\end{theorem}
\begin{proof}
\emph{Theorem 1} ensures that this LP produces a valid solution. Since the objective (\ref{operation 8}) is a minimization function and fractional weights are allowed, it follows that the solution produced by LP is optimal.

For edge failures, the LP formulation \ref{sec:offline-LP} has to be modified by removing the failed edges, adding constraints for the current weight assignments and adding constraints for new set of demands until the next edge failure. 
\end{proof}

\subsection{Primal-Dual}
\label{sec:offline-PD}
Primal-Dual algorithms \cite{distributed-comb-opt} are used for a certain class of optimization problems that involve minimization or maximization of an objective function where there are a finite number of feasible solutions available at each step. These algorithms are based on constructing a dual which is solved in conjunction with the primal. It is used to derive intuitions about the nature of the solution that are implicit in the primal.

Consider the dual \cite{lp,algorithms} of the LP formulation in section \ref{sec:offline-LP}. Let $y_i$ be the dual variables corresponding to producers $i \in P$ (\ref{operation 9}) and $z_j$ be the dual variables corresponding to the consumers $j \in C$ (\ref{operation 10}) then the corresponding dual is,

Objective function:

Maximize:
\begin{equation}
\displaystyle \sum_{i \in P} y_i  \cdot  R_i - \sum_{j \in C} z_j  \cdot  M_j, \; y_i \ge 0, \; z_j \ge 0
\label{operation 11}
\end{equation} 

Constraints:
\begin{equation}
\displaystyle y_i - z_j \le d_{ij} , \forall i \in P, j \in C
\label{operation 12}
\end{equation}

Equation (\ref{operation 12}) suggests that the potential difference between producers and consumers can be atmost equal to $d_{ij}$. This will be referred to as dual \emph{potential limit constraint}.

Note that by complementary slackness conditions,
\begin{equation}
\displaystyle w_{ij} > 0 \iff y_i - z_j = d_{ij}
\label{operation 13}
\end{equation}

By complementary slackness, weights are allocated ($w_{ij} > 0$) on edges $e_{ij}, \; i \in P$, $j \in C$ if the potential difference between producer $y_i$ and consumer $z_j$ becomes equal to $d_{ij}$ and vice-versa by (\ref{operation 13}).

Let $T(y_i)$ be the set of tight constraints for $y_i$ such that,
\begin{equation}
\displaystyle T(y_i) = \{(i,j): y_i - z_j = d_{ij}\}
\label{operation 14}
\end{equation}

Let $S(y_i)$ be the set of slack constraints for $y_i$ such that,
\begin{equation}
\displaystyle S(y_i) = \{(i,j): y_i - z_j < d_{ij}\}
\label{operation 15}
\end{equation}

Consider an \emph{unit benefit} function of $y_i$ which measures the increase in dual objective \ref{operation 11},
\begin{equation}
\displaystyle B(y_i) = R_i - \sum_{j: (i,j) \in T(y_i)} M_j
\label{operation 17}
\end{equation}

\begin{algorithm}
\caption{Primal-Dual algorithm for \emph{Offline Assignment}}
\label{alg1}
\begin{algorithmic}
\STATE $y_i \gets 0, \forall i \in P$
\STATE $z_j \gets 0, \forall j \in C$
\STATE $T(y_i) \gets \emptyset$
\STATE $S(y_i) \gets \emptyset$
\WHILE{ $\exists y_i: B(y_i) \geq 0$}
\STATE $y_i: Max_{B(y_i)}$
\STATE $\delta_1 = \{ Min_{d_{ij} - (y_i - z_j)} \;| \; (i,j) \in S(y_i) \}$
\STATE $y_i \gets y_i + \delta_1$
\STATE $z_j \leftarrow z_j + \delta_1, \forall j: (i,j) \in T(y_i)$
\STATE $T(y_i) = \{(i,j) \; | \; y_i - z_j = d_{ij}, \forall j \in C\}$
\STATE  $\delta_2 = Min_{j \in C}(M_j - \sum_{i \in P} w_{ij: (i,j) \in T(y_i)})$
\STATE $w_{ij} \gets w_{ij} + \delta_2$
\ENDWHILE
\end{algorithmic}
\end{algorithm}

Initializing $y_i \gets 0, \forall i \in P$ and $z_j \gets 0, \forall j \in C$ produces a dual feasible solution as the dual \emph{potential limit constraint} (\ref{operation 12}) is satisfied. The primal-dual algorithm chooses the $y_i$ with the highest benefit function $B(y_i)$ at each step to maximize the increase in value of dual objective function. It then chooses the constraint that is closest to becoming tight and increases the value of $y_i$ by the amount that is needed to make this constraint tight, $\delta_1 =  Min_{(i,j) \in S(y_i)} (d_{ij} - (y_i - z_j))$.

For the set of constraints T that are already tight the corresponding $z_j$ are also increased by $\delta_1$ to maintain tightness. For the constraints that just became tight, the corresponding $w_{ij}$ are increased to the value of the least available capacity amongst all consumers,

\begin{equation}
Min_{j \in C}(M_j - \sum_{i \in P} w_{ij: (i,j) \in T(y_i)})
\label{operation 18}
\end{equation}

\begin{theorem}  {\bf (Optimality of Algorithm \ref{alg1})} The Primal-Dual Algorithm \ref{alg1} reaches the optimal solution for the \emph{assignment with reallocation} problem in section \ref{sec:problem-definition} when it is not possible to increase the potentials in the corresponding dual any further. 
\end{theorem}
\begin{IEEEproof} 
This algorithm always produces a dual feasible solution as the dual constraints in \emph{potential limit constraint} (\ref{operation 12}) are always satisfied by definition of Algorithm \ref{alg1}. The primal \emph{consumer capacity constraints} in (\ref{operation 10}) are always satisfied from the way we increase $w_{ij}$ from (\ref{operation 18}).
When it is not possible to increase the value of dual objective function the benefit function (\ref{operation 17}) has a negative value,
\begin{equation}
B(y_i) < 0 \implies R_i < \sum_{j: (i,j) \in T(y_i)} M_j, \; \forall i \in P
\label{operation 19}
\end{equation}
From (\ref{operation 18}), we know that,
\begin{equation}
\sum_{j: (i,j) \in T(y_i)} M_j = \sum_{j:(i,j) \in T(y_i)} w_{ij}
\label{operation 20}
\end{equation}

Using (\ref{operation 19}) and (\ref{operation 20}) we infer that the demands have been met. This means that primal constraints in (\ref{operation 9}) have been satisfied and the dual constraints are always satisfied. Thus the solution is optimal.

Complementary slackness (\ref{operation 13}) is satisfied as we only increase $w_{ij}$ when the dual constraint is tight. This means that the primal is optimal.
\end{IEEEproof}

\begin{theorem} {\bf (Complexity of Algorithm \ref{alg1})} The Primal-Dual Algorithm \ref{alg1} takes $O(n^3), \; n = |P + C|$ time to complete.
\end{theorem}
\begin{IEEEproof} 
For each producer, it takes $O(|P|)$ comparisions to calculate the dual variable $y_i$ with the maximum unit benefit function,  $O(|P| \ cdot |C|)$ comparisions (which is equal to the number of dual constraints) to calculate $delta_1$  and $O(|C|)$ comparisions (which finds the minimum amongst all consumers) to find the $delta_2$. Thus it takes $O(|P|+|P| \cdot |C|+|C|)$ to execute the while loop in Algorithm \ref {alg1}. For $|P|$ producers it takes $O(|P|^2+|P|^2 \cdot |C|+|P| \cdot |C|) = O (||P|^2 \cdot |C|) = O(n^3), \; n = |P + C|$ comparisions.
\end{IEEEproof}

\section{Online Algorithms}
\label{sec:online-algorithms}

Online algorithms are used for solving problems where the input is received incrementally and partial decisions have to be made at each step. Competitive ratio is used to measure the performance of online algorithm as compared to the optimal offline algorithm that knows the entire input. A study of how randomization can be used to improve the competitiveness of online algorithms is presented in \cite{power-of-randomization-online}.

An $\alpha$-competitive online algorithm ALG is defined as follows with respect to an optimal offline algorithm OPT, for a problem \emph{P} where, \emph{I} is an instance of the problem,
\begin{equation}
cost(ALG(I)) \le\alpha \cdot cost(OPT(I)) + \beta, \; \forall I \in P
\label{operation 36}
\end{equation}

In equation (\ref{operation 36}), $\alpha$ is called the \emph{competitive ratio} and $\beta$ can be considered as the \emph{startup cost} of the algorithm. This paper assumes a startup cost of zero, $\beta = 0$.

\subsection{Assignment without reallocation}
\label{sec:solution-2}

\begin{algorithm}
\caption{Primal-Dual algorithm for \emph{Assignment without reallocation}}
\label{alg2}
\begin{algorithmic}
\STATE $y_i \gets 0, \forall i \in P$
\STATE $z_j \gets 0, \forall j \in C$
\STATE $T(y_i) \gets \emptyset$
\STATE $S(y_i) \gets \emptyset$
\WHILE{ producer demand}
\STATE $\delta_1 = \{ d_{ij} - (y_i - z_j)\;|\;(i,j) = Random(S(y_i)) \}$
\STATE $y_i \gets y_i + \delta_1$
\STATE $z_j \leftarrow z_j + \delta_1, \forall j: (i,j) \in T(y_i)$
\STATE $T(y_i) = \{(i,j) \; | \; y_i - z_j = d_{ij}, \forall j \in C\}$
\STATE  $\delta_2 = min_{j \in C}(M_j - \sum_{i \in P} w_{ij: (i,j) \in T(y_i)})$
\STATE $w_{ij} \gets w_{ij} + \delta_2$
\ENDWHILE
\end{algorithmic}
\end{algorithm}

The online adversary \cite{competitive-analysis} produces a sequence of demands that decreases the performance of deterministic algorithms. For a greedy algorithm that selects the minimum cost edge the online adversary produces demands in a non-decreasing order of magnitude denoted $R_{adv}$  such that, the highest demands $R_{MAX}$ are assigned to edges with the highest distances $d_{MAX}$ to maximize the value of objective \ref{operation 8}:

$R_{adv}=R_{MIN} \cdots R_{MAX}$. 

In response to the service request sequence $R_{adv}$ the Algorithm \ref{alg2} produces a random assignment of weights on edges. After selecting a random edge for weight assignment the corresponding dual variable $y_i$ is increased to $d_{ij}$ to satisfy the complementary slackness condition in (\ref{operation 13}).

Cost (\ref{operation 8}) of the solution produced by the Algorithm \ref{alg2} is,
\begin{equation}
Cost(ALG2) = \sum R_i  \cdot  E[d_{e_{ij}}]
\label{operation 21}
\end{equation}

\begin{equation}
\begin{array}{l}
\displaystyle E[d_{e_{ij}}] = \sum_{i, j=1}^{|C|} \frac{d_{e_{ij}}}{j} \\
\displaystyle \;\;\;\;\;\;\;\;\;\;\;\; \le d_{e_{MAX}} \cdot \ln |C| \;\;\;\; (\mbox{where $H_n$ is the nth Harmonic})
\label{operation 22}
\end{array}
\end{equation}

At the first iteration, Algorithm \ref{alg2} selects an edge from all the available edges, which is $n = |C|$. By definition of Algorithm \ref{alg2} after selecting an edge, the weight on the edge $w_{ij}$ is increased until the capacity of consumer j, $M_j$ is reached. So, in the next iteration, one amongst the remaining $n-1$ edges has to be selected. This process is repeated until all consumers are saturated.

Substituting (\ref{operation 22}) in (\ref{operation 21}),
\begin{equation}
Cost(ALG2) \le d_{e_{MAX}} \cdot \ln |C| \cdot \sum R_i
\label{operation 23}
\end{equation}

The optimal offline solution assigns the lower valued demands to the higher cost edges and saves the lower cost edges for the higher valued demands that arrive later in the sequence. Cost of the solution produced by the optimal offline algorithm is,

\begin{equation}
\begin{array}{l}
\displaystyle Cost(OPT) = d_{e_{MIN}}  \cdot  R_{MAX} +  \cdots + d_{e_{MAX}}  \cdot  R_{MIN} \\
\displaystyle \;\;\;\;\;\;\;\;\;\;\;\;\;\;\;\;\;\;\;\; \ge d_{e_{MIN}}  \cdot \sum R_i\\
\label{operation 26}
\end{array}
\end{equation}

From (\ref{operation 23}) and (\ref{operation 26}), the competitive ratio ((\ref{operation 36}) in \ref{sec:online-algorithms}) for Algorithm \ref{alg2} is,

\begin{equation}
\begin{array}{l}
\displaystyle \alpha \le  \frac{d_{e_{MAX}}}{d_{e_{MIN}}} \cdot \ln |C|
\end{array}
\label{operation 24}
\end{equation}

This algorithm also runs in $O(n^3)$ time similar to Algorithm \ref{alg1} as although choosing a random edge takes constant time, the running time of the algorithm is bounded by the operation that calculates the tight constraints $T(y_i)$.

\subsection{Assignment with varying edge distances}
\label{sec:solution-3}
\begin{algorithm}
\caption{Primal-Dual algorithm for \emph{Assignment with varying edge distances}}
\label{alg3}
\begin{algorithmic}
\STATE $y_i \gets 0, \forall i \in P$
\STATE $z_j \gets 0, \forall j \in C$
\STATE $T(y_i) \gets \emptyset$
\STATE $S(y_i) \gets \emptyset$
\WHILE{producer demand}
\STATE $\delta_1 = \{ d_{ij} - (y_i - z_j)\;|\;(i,j) = Random(S(y_i)) \}$
\STATE $y_i \gets y_i + \delta_1$
\IF {\mbox{edge distance changes}}
\STATE $z_j \gets z_j + \delta_1 - (\overline{d_{ij}} - d_{ij}), \forall j: (i,j) \in T(y_i)$
\ELSE
\STATE $z_j \gets z_j + \delta_1, \forall j: (i,j) \in T(y_i)$
\ENDIF
\STATE $T(y_i) = \{(i,j) \; | \; y_i - z_j = d_{ij}, \forall j \in C\}$
\STATE  $\delta_2 = Min_{j \in C}(M_j - \sum_{i \in P} w_{ij: (i,j) \in T(y_i)})$
\STATE $w_{ij} \gets w_{ij} + \delta_2$
\ENDWHILE
\end{algorithmic}
\end{algorithm}

This section only considers cases where $d_{diff} > 0$ as these cases form an upper bound on the competitive ratio. As for cases where $d_{diff} \le 0$ the cost of the primal objective \ref{operation 8} either remains same or moves towards the optimal solution.

In Algorithm \ref{alg3} the primal-dual algorithm accommodates the varying edge distances by adjusting the potentials in \emph{potential limit constraint} (\ref{operation 12}) of producers $y_i$ and consumers $z_j$ in order to maintain the tight constraints T (\ref{operation 14}). For the new constraints that become tight $\overline{T}$, the corresponding primal variables $w_{ij}$ are raised as much as possible.

The online adversary targets the edge with highest weight assignment, $Max_{w_{ij}}$ and increases its distance by a certain amount: $d_{diff} = \overline{d_{ij}} - d_{ij}$. Let the existing weight assignment be $A_{ALG3} = \{e_1 \gets R_1,  \cdots, e_n \gets R_n \}$. Say the adversary targets the edge with the highest weight assignment $Max_{w_{ij}}$ and increases the distance of this edge from $d_{ij}$ to $\overline{d_{ij}}$.

At each iteration of primal-dual Algorithm \ref{alg3} the expected increase in cost due to varying edge distances $d_{diff}$ is,
\begin{equation}
\begin{array}{l}
\displaystyle E[Cost_{ALG3}] = Max_{ALG3}(w_{e_{ij}})  \cdot  d_{diff}
\end{array}
\label{operation 25}
\end{equation}

At each iteration of OPT \ref{sec:offline-LP} the maximum increase in cost due to varying edge distances is,
\begin{equation}
\begin{array}{l}
\displaystyle E[Cost_{OPT}] = Max_{OPT}(w_{e_{ij}}) \cdot d_{diff}
\end{array}
\label{operation 28}
\end{equation}

Competitive ratio of Algorithm \ref{alg7} is the same as the competitive ratio \ref{operation 24} for \emph{assignment without reallocation} \ref{sec:problem-2} as maximum weight assignments for the Optimal, OPT and Algorithm \ref{alg3} are equal to the value of maximum demand ($Max_{ALG3}(w_{e_{ij}})  = Max_{OPT}(w_{e_{ij}}) = Max(R_k)$).

This algorithm also runs in $O(n^2)$ time similar to Algorithm \ref{alg1}.

\subsection{Assignment with node addition / failure or attribute changes}
\label{sec:solution-4}
\begin{algorithm}
\caption{Primal-Dual algorithm for \emph{Assignment with producer failures}}
\label{alg4}
\begin{algorithmic}
\STATE $y_i \gets 0, \forall i \in P$
\STATE $z_j \gets 0, \forall j \in C$
\STATE $T(y_i) \gets \emptyset$
\STATE $S(y_i) \gets \emptyset$
\WHILE{producer demand}
\STATE $\delta_1 = \{ d_{ij} - (y_i - z_j)\;|\;(i,j) = Random(S(y_i)) \}$
\STATE $y_i \gets y_i + \delta_1$
\STATE $z_j \gets z_j + \delta_1 - (\overline{d_{ij}} - d_{ij}), \forall j: (i,j) \in T(y_i)$
\STATE $T(y_i) = \{(i,j) \; | \; y_i - z_j = d_{ij}, \forall j \in C\}$
\IF {producer i fails}
\STATE  $C_i = \{j\; | \; w_{ij} > 0, \forall j \in C\}$
\FOR{$j \in C_i$}
\STATE $w_{ij} = 0, \forall j \in C_i$
\STATE $T(y_i) \gets T(y_i) \setminus (i,j)$
\ENDFOR
\ENDIF
\STATE  $\delta_2 = Min_{j \in C}(M_j - \sum_{i \in P} w_{ij: (i,j) \in T(y_i)})$
\STATE $w_{ij} \gets w_{ij} + \delta_2$
\ENDWHILE
\end{algorithmic}
\end{algorithm}

In problem \ref{sec:problem-5} the producers / consumers may go down or come up or their demands / capacities can change over time.

\subsubsection{Producer failure}
\label{sec:producer-failure}

When a producer $i \in P$ goes down, the Algorithm \ref{alg4} first calculates $C_i = \{j \; | \; w_{ij} > 0, \forall j \in C\}$. Then the available capacity $M_j$ of each consumer on which producer $i$ demands were assigned is increased by $w_{ij}$, $M_j \gets M_j + w_{ij}, \forall j \in C_i$. After this, the corresponding weight assignments are set to zero, $w_{ij} = 0, \forall j \in C_i$.

The competitive ratio for Algorithm \ref{alg4} is equal to the competitive ratio for \emph{assignment without reallocation} (\ref{operation 24}) as the operations corresponding to a producer failure do not affect the competitive ratio.

Additional steps corresponding to the failure of the producer, add $O(|C|)$ to the time complexity for checking if each consumer contains weight assignments corresponding to producer $i \in P$. Overall time complexity is $O(n^2+|C|) = O(n^2)$ as $|C| < |V| = n$.

It is to be noted that producers are added online by definition \ref{sec:problem-definition} and the demands generated by them cannot change over time.

\subsubsection{Consumer addition}
\label{sec:Consumer-addition}
\begin{algorithm}
\caption{Primal-Dual algorithm for \emph{Assignment with consumer addition}}
\label{alg5}
\begin{algorithmic}
\STATE $y_i \gets 0, \forall i \in P$
\STATE $z_j \gets 0, \forall j \in C$
\STATE $T(y_i) \gets \emptyset$
\STATE $S(y_i) \gets \emptyset$
\WHILE{producer demand}
\STATE $\delta_1 = \{ d_{ij} - (y_i - z_j)\;|\;(i,j) = Random(S(y_i)) \}$
\STATE $y_i \gets y_i + \delta_1$
\STATE $z_j \gets z_j + \delta_1 - \overline{(d_{ij}} - d_{ij}), \forall j: (i,j) \in T(y_i)$
\STATE $T(y_i) = \{(i,j) \; | \; y_i - z_j = d_{ij}, \forall j \in C\}$
\IF {consumer j is added}
\FOR{$i \in P$}
\STATE $w_{ij}: Max_{d_{ij}}, \forall  w_{ij} > 0$
\STATE $\overline{w_{ij}} = w_{ij}$
\STATE $w_{ij} = 0$
\STATE $T(y_i) \gets T(y_i) \cup \overline{(i,j)} \setminus (i,j)$
\ENDFOR
\ENDIF
\STATE  $\delta_2 = Min_{j \in C}(M_j - \sum_{i \in P} w_{ij: (i,j) \in T(y_i)})$
\STATE $w_{ij} \gets w_{ij} + \delta_2$
\ENDWHILE
\end{algorithmic}
\end{algorithm}

New consumers will bring with them a new set of edges corresponding to each producer. This can only decrease the cost of existing solution as each producer now has one more edge to choose from for weight assignment. 

If distance of this newly added edge $\overline{d_{ij}}$ is equal to or more than distances of all the existing edges for which $w_{ij} > 0$ then the weight assignment and cost remains the same. Otherwise the weight is transferred from the current edge $e_{ij}$ with the highest distance to the newly added edge $\overline{e_{ij}}$.

The competitive ratio for the primal-dual Algorithm \ref{alg5} is equal to competitive ratio (\ref{operation 24}) of \emph{assignment without reallocation} problem \ref{sec:problem-2} as the cost can only decrease due to newly added edges. For each $i \in P$ it takes constant time to compare the distance of the newly added edge to the edge on which the weight is currently assigned. This take an additional time of $O(|P|)$. The overall time complexity is $O(n^2+|P|) = O(n^2), |P| < n$.

\subsubsection{Consumer capacities decrease}
\label{sec:consumer-capacity-decrease}

\begin{algorithm}
\caption{Primal-Dual algorithm for \emph{Assignment with consumer capacity decrease}}
\label{alg6}
\begin{algorithmic}
\STATE $y_i \gets 0, \forall i \in P$
\STATE $z_j \gets 0, \forall j \in C$
\STATE $T(y_i) \gets \emptyset$
\STATE $S(y_i) \gets \emptyset$
\STATE {\bf A:}
\WHILE{Producer demand}
\STATE $\delta_1 = \{ d_{ij} - (y_i - z_j)\;|\;(i,j) = Random(S(y_i)) \}$
\STATE $y_i \gets y_i + \delta_1$
\STATE $z_j \gets z_j + \delta_1, \forall (i,j) \in T(y_i)$
\STATE $T(y_i) = \{(i,j) \; | \; y_i - z_j = d_{ij}, \forall j \in C\}$
\IF {$(\overline{M_{j}} - M_j) < 0$}
\STATE $P_j = \{i \; | \; w_{ij} > 0, \forall i \in P \}$
\STATE $i = Min(d_{ij}), \; i \in P_j$
\STATE $w_{ij} \gets w_{ij} - |\overline{M_j} - M_j|$
\STATE $w_{\overline{i}j} \gets w_{\overline{i}j} + |\overline{M_j} - M_j|, \; \overline{i}: = Min(d_{ij}), i \neq \overline{i}$
\ENDIF
\STATE  $\delta_2 = min_{j \in C}(M_j - \sum_{i \in P} w_{ij: (i,j) \in T(y_i)})$
\STATE $w_{ij} \gets w_{ij} + \delta_2$
\STATE generate demand of value $|\overline{M_j} - M_j|$ for producer i
\STATE {\bf goto} A
\ENDWHILE
\end{algorithmic}
\end{algorithm}

The primal-dual Algorithm \ref{alg6} selects the producer with the minimum distance edge ($i = Min(d_{ij}), \; i \in P_j$), decreases the weight assigned on edge $d_{ij}$ by the consumer's residual weight $M_{res}$ and assigns $M_{res}$ on the edge with the next lowest distance. The increase in cost of the dual objective due to this operation is:

\begin{equation}
\displaystyle E[Cost_{ALG6}] = E[Cost_{OPT}] = M_{Res} \cdot (d_{\overline{i}j} - d_{ij})
\label{operation 30}
\end{equation}

The Competitive ratio of Algorithm \ref{alg6} is the same as the competitive ratio \ref{operation 24} for \emph{assignment without reallocation} \ref{sec:problem-2} as the optimal algorithm does the same in case of consumer capacity decrease.

Time complexity of Algorithm \ref{alg6} is the time required to service external demands, time required to find the producers in $P_j$ and the time to service internal demands $O(n^2_{external} + n^2_{internal} +|P|) = O(n^2), |P| <n$

\subsubsection{Consumer failure}
\label{sec:consumer-failure}
\begin{algorithm}
\caption{Primal-Dual algorithm for \emph{Assignment with consumer failures}}
\label{alg7}
\begin{algorithmic}
\STATE $y_i \gets 0, \forall i \in P$
\STATE $z_j \gets 0, \forall j \in C$
\STATE $T(y_i) \gets \emptyset$
\STATE $S(y_i) \gets \emptyset$
\STATE {\bf A:}
\WHILE{producer demand}
\STATE $\delta_1 = \{ d_{ij} - (y_i - z_j)\;|\;(i,j) = Random(S(y_i)) \}$
\STATE $y_i \gets y_i + \delta_1$
\STATE $z_j \gets z_j + \delta_1 - (\overline{d_{ij}} - d_{ij}), \forall (i,j) \in T(y_i)$
\STATE $T(y_i) = \{(i,j) \; | \; y_i - z_j = d_{ij}, \forall j \in C\}$
\IF {consumer j fails}
\STATE $P_j = \{i \; | \; w_{ij} > 0, \forall i \in P \}$
\FOR{$i \in P_j$}
\STATE generate demand of value $w_{ij}$ for producer i
\STATE {\bf goto} A
\ENDFOR
\ENDIF
\STATE  $\delta_2 = Min_{j \in C}(M_j - \sum_{i \in P} w_{ij: (i,j) \in T(y_i)})$
\STATE $w_{ij} \gets w_{ij} + \delta_2$
\ENDWHILE
\end{algorithmic}
\end{algorithm}

Consumer failure can invalidate a current assignment. In this case the weight assignments corresponding to the producers that have weights allocated on the failed consumer $j$, $P_j = \{i \; | \; w_{ij} > 0, \forall i \in P \}$ are invalidated. This generates internal \emph{residual demand} corresponding to each producer $i \in P_j$. This is handled as a regular demand by the Algorithm \ref{alg7}.

The cost of overall solution may increase or decrease depending upon the edge selected for reassignment of weights. If the distance of edge selected is higher than the distance of the edge on which the weight was assigned initially, $\overline{d_{ij}} - d_{ij} > 0$ then the cost increases. If $\overline{d_{ij}} - d_{ij} < 0$ the cost decreases. The cost remains the same if $\overline{d_{ij}} = d_{ij}$. 

Optimal algorithm will choose the edge with minimum cost available for assigning residual weights. This edge can be found by going through all available edges $|C|$ for each producer. Thus the Competitive ratio of Algorithm \ref{alg7} is the same as the competitive ratio \ref{operation 24} for \emph{assignment without reallocation} \ref{sec:problem-2}.

Time complexity of Algorithm \ref{alg7} is equal to that of \ref{sec:consumer-capacity-decrease} as this is a special case of \emph{assignment with consumer capacity decrease} \ref{sec:consumer-capacity-decrease} where consumer capacity is set to zero.

\subsection{Offline Assignment with multiple producer requests}
\label{sec:problem-6}
The primal constraint \ref{operation 9} is extended as follows

\begin{equation}
\sum_{j \in C} w_{ij}(t) \ge \sum_{t} R_i(t), \forall i \in P, \; t \in T
\label{operation 31}
\end{equation}

The unit benefit function(UBF) \ref{operation 17}  is now calculated for each producer demand as follows:

\begin{equation}
\displaystyle B(y_{it}) = R_{it} - \sum_{j: (i,j) \in T(y_i)} M_j
\label{operation 32}
\end{equation}

Instead of increasing $y_i$ (dual variable corresponding to the producer $i \in P$) by the entire amount needed to make it tight ($\delta_1 = d_{ij} - (y_i - z_j)$), we only increase by the amount proportional to its share in the total producer demand ($\delta_1 *(R_{it}/\sum_{t}R{it})$). This ensures that the edge $e_{ij}$ does not become tight until the producer $j \in C$ is saturated.

\begin{algorithm}
\caption{Primal-Dual algorithm for \emph{Offline Assignment with multiple producer demands}}
\label{alg8}
\begin{algorithmic}
\STATE $y_i \gets 0, \forall i \in P$
\STATE $z_j \gets 0, \forall j \in C$
\STATE $T(y_i) \gets \emptyset$
\STATE $S(y_i) \gets \emptyset$
\WHILE{ $\exists y_i: B(y_{it}) \geq 0$}
\STATE $y_i: Max_{B(y_{it})}$
\STATE $\delta_1 = \{ Min_{d_{ij} - (y_i - z_j)} \;| \; (i,j) \in S(y_i) \}$
\STATE $y_i \gets y_i + \delta_1 *(R_{it}/(M_j - \sum_{(i,j) \in T} w_{ij})$
\STATE $z_j \leftarrow z_j + \delta_1, \forall j: (i,j) \in T(y_i)$
\STATE $T(y_i) = \{(i,j) \; | \; y_i - z_j = d_{ij}, \forall j \in C\}$
\STATE  $\delta_2 = Min_{j \in C}(M_j - \sum_{i \in P} w_{ij: (i,j) \in T(y_i)})$
\STATE $w_{ij} \gets w_{ij} + \delta_2$
\ENDWHILE
\end{algorithmic}
\end{algorithm}

Note that we assume the consumers gets saturated exactly although there could be a producer demand that can only be partly allocated on a given consumer and the remaining part has to be allocated on a different consumer.

\section{Conclusion}
\label{sec:conclusion}
Variants of the online assignment problem defined in section \ref{sec:problem-definition} can be solved using efficient primal-dual algorithms. Implementing this theoretical approach will improve the performance of automated storage management schemes used in distributed systems.

\section{Acknowledgement}
\label{sec:acknowledgement}
The author would like to thank Prof. Manindra Agrawal and Prof. Sumit Ganguly at the Department of Computer Science and Engineering at the Indian Institute of Technology (IIT), Kanpur for their valuable suggestions and guidance.


\begin{thebibliography}{99}
\bibitem{hungarian-algo}
Kuhn HW,
The Hungarian method for the assignment problem,
Naval Research Logistics, Quarterly {\bf 2}, 1955, 83–97.
\bibitem{incr-assign}
Toroslu IH, ¨Uc¸ Oluk G,
Incremental assignment problem,
Information Sciences, {\bf 177}, 6 (March 2007), 1523–1529.
\bibitem{dynamic-hungarian-algo}
Mills-Tettey, G. A., Stentz, A. T., and Dias, M. B.,
The dynamic Hungarian algorithm for the assignment problem with changing costs,
Tech. rep. CMU-RI-TR-07-27, Robotics Institute, 2007.
\bibitem{online-matching}
B. Fuchs, W. Hochstattler, and W. Kern,
Online matching on a line,
Theoretical Computer Science {\bf 332} (1-3), 2005.
\bibitem{randomized-online}
Adam Meyerson, Akash Nanavati, Laura Poplawski,
Randomized online algorithms for minimum metric bipartite matching,
Proceedings of the seventeenth annual ACM-SIAM symposium on Discrete algorithms, 954 - 959, 2006.
\bibitem{online-weighted-bip-match}
S. Khuller, S. Mitchell, and V. Vazirani,
On-line algorithms for weighted bipartite matching and stable marriages,
Theory of Computer Science {\bf 127(2)}, 1994.
\bibitem{min-flow-related}
Garg N, Kumar A,
Better algorithms for minimizing average flow-time on related machines,
Automata, Languages and Programming, Pt. 1, Lecture Notes in Computer Science, {\bf 4051}, 181-190, 2006.
\bibitem{bip-network-flow}
Ahuja RK, Orlin JB, Stein C, et al.,
Improved algorithms for bipartite network flow,
SIAM Journal on Computing, {\bf 23}, 5, 906-933, Oct 1994.
\bibitem{dynamic-matching-bipartite}
Brodal GS, Georgiadis L, Hansen KA, Katriel I,
Dynamic Matchings in Convex Bipartite Graphs,
Mathematical Foundations of Computer Science, Lecture Notes in Computer Science, {\bf 4708} , 406-417, 2007.
\bibitem{interior-point-methods}
Goldberg AV, Plotkin SA, Shmoys DB, et al.,
Using interior-point methods for fast parallel algorithms for bipartite matching and related problems,
SIAM Journal on Computing, {\bf 21}, 1, 140-150, Feb 1992.
\bibitem{lp-onlline}
Correa JR, Wagner MR, 
LP-based online scheduling: From single to parallel machines,
Integer programming and combinatorial optimization, proceedings, lecture notes in computer science, { \bf 3509}, 196-209, 2005. 
\bibitem{job-assign-scalable}
Amir Y, Awerbuch B, Barak A, Borgstrom RS, Keren A, 
An opportunity cost approach for job assignment in a scalable computing cluster,
IEEE transactions on parallel and distributed systems { \bf 11}, 7, 760-768, Jul 2000.
\bibitem{dynamic-assignment}
Spivey MZ, Powell WB, The dynamic assignment problem,
Transportation Science 38, {\bf 4}, November 2004, 399–419.
\bibitem{rand-dyn-graphs}
Monika R. Henzinger, Valerie King,
Randomized fully dynamic graph algorithms with polylogarithmic time per operation,
Journal of the ACM (JACM) archive, { \bf 46}, 4, 502 - 516, July 1999.
\bibitem{emp-dyn-graph-algos}
David Alberts, Giuseppe Cattaneo, Giuseppe F. Italiano,
An empirical study of dynamic graph algorithms,
Journal of Experimental Algorithmics (JEA) archive, { \bf 2}, article 5, 1997.
\bibitem{multiflow-dyn-graph-algos}
Aleksander Madry,
Faster approximation schemes for fractional multicommodity flow problems via dynamic graph algorithms,
Annual ACM Symposium on Theory of Computing archive, { \bf 2}, Session 2A, 121-130, 2010.
\bibitem{survey-distributed-computing}
Waldspurger CA, Hogg T, Huberman BA, et al., 
Spawn - A Distributed Computational Economy,
IEEE Transactions on Software Engineering, {\bf 18} , no. 2, 103-117, Feb {1992} .
\bibitem{fairness-alloc}
Baruah SK, Cohen NK, Plaxton CG, Varvel DA, 
Proportionate progress- A notion of fairness in producer allocation,
Algorithmica, {\bf 15}, 6, 600-625, Jun 1996 .
\bibitem{balanced-alloc}
Azar Y, Broder AZ, Karlin AR, Upfal E, 
Balanced allocations,
SIAM Journal of Computing, {\bf 29},1, 180-200, Sep 1999 .
\bibitem{power-of-randomization-online}
S. Ben-David, A. Borodin, R. Karp, G. Tardos, A. Wigderson, 
On the Power of Randomization in On-line Algorithms,
Algorithmica, vol. { \bf 11}, no. 1, pp. 2-14, 1994.
\bibitem{chromatic-sums}
Bar-Noy A, Bellare M, Halldorsson MM, Shachnai H, Tamir T, 
On chromatic sums and distributed producer allocation,
Information and Computation {\bf 140}, 2, 183-202, Feb 1 1998.
\bibitem{constraint-driven-scheduling}
Kuchcinski K, Constraints-driven scheduling and producer assignment
ACM transactions on design automation of electronic systems, { \bf 8}, 3, 355-383, Jul 2003.
\bibitem{distributed-comb-opt} 
Aydin ME, Fogarty TC, 
A distributed evolutionary simulated annealing algorithm for combinatorial optimization problems
Journal of Heuristics, { \bf 10}, 3, 269-292, May 2004.
\bibitem{cloud-computing}
Buyya R, Yeo CS, Venugopal S, Broberg J, Brandic I, 
Cloud computing and emerging IT platforms: Vision, hype, and reality for delivering computing as the 5th utility,
Future Generation Computer Systems-The International Journal of Grid Computing-Theory Methods And Applications, {\bf 25}, 6, 599-616, Jun 2009.
\bibitem{related-load-balancing}
Berman P, Charikar M, Karpinski M, 
On-line load balancing for related machines,
Journal of Algorithms, {\bf 35}, 1, 108-121, Apr 2000.
\bibitem{online-load-balancing}
Phillips S, Westbrook J, On-line load balancing and network flow
Algorithmica, {\bf 21}, 3, 245-261, Jul 1998.
\bibitem{k-server-offline}
Tomislav Rudec, Alfonzo Baumgartner, Robert Manger, 
A fast implementation of the optimal off-line algorithm for solving the k-server problem, 
Math. Commun., {\bf14}, No. 1, pp. 119-134 (2009).
\bibitem{k-server-lower-bounds}
Irene Fink, Sven O. Krumke, Stephan Westphal, 
New lower bounds for online k-server routing problems, 
Information Processing Letters {\bf109} (2009) 563–567.
\bibitem{k-server-competitive}
Mark S. Manasse, Lyle A. McGeoch, Daniel D. Sleator, 
Competitive algorithms for server problems, 
Journal of Algorithms, {\bf11} , Issue 2 (June 1990), 208 - 230.
\bibitem{k-server-decomposition}
Steven S. Seiden, 
A General Decomposition Theorem for the k-Server Problem, 
Information and Computation {\bf174}, 193–202 (2002).
\bibitem{k-server-survey}
Vincenzo Bonifaci, Leen Stougie, 
Online k-Server Routing Problems, 
Theory Comput Syst (2009) {\bf45}, 470–485.
\bibitem{k-server-randomized}
Yair Bartal, Manor Mendel, 
Randomized k-server algorithms for growth-rate bounded graphs, 
Journal of Algorithms {\bf55} (2005) 192–202.
\bibitem{hierearchical-load-balancing}
Bar-Noy A, Freund A, Naor JS, 
On-line load balancing in a hierarchical server topology,
SIAM Journal on Computing, { \bf 31}, 2, 527-549, Oct 2001.
\bibitem{lp}
Hadley, G. Linear Programming. Addison-Wesley, Reading, Mass., 1962, p. 221-318.
\bibitem{algorithms}
Sanjoy Dasgupta, Christos Papadimitriou, Umesh Vazirani, Algorithms, McGraw-Hill, 1st edition, 2006, p. 201-246.
\bibitem{competitive-analysis}
Allan Borodin, Ran El-Yaniv, Online computation and competitive analysis, Cambridge University Press, New York, 1998. 
\bibitem{factor-graphs}
Frank R. Kschischang, Brendan J. Frey, Hans-Andrea Loeliger,
Factor Graphs and the Sum-Product Algorithm,
IEEE Transactions on Information Theory, {\bf 47}, NO. 2, Feb 2001, p. 498-519
Allan Borodin, Ran El-Yaniv, Online computation and competitive analysis, Cambridge University Press, New York, 1998. 
\bibitem{vmware-drs}
\href{http://www.vmware.com/files/pdf/VMware-DynamicStorageProv-WP-EN.pdf}{Dynamic Storage Provisioning},
VMware, Nov 18, 2009.
\bibitem{vmware-scale-storage}
\href{http://www.vmware.com/files/pdf/scalable_storage_performance.pdf}{Scalable Storage Performance},
VMware, Jun 5, 2008.
\bibitem{lp-solve}
lpsolve, Open source (Mixed-Integer) Linear Programming system, 
Multi-platform, pure ANSI C / POSIX source code, Lex / Yacc based parsing,
Version 5.1.0.0 dated 1 May 2004,
Michel Berkelaar, Kjell Eikland, Peter Notebaert,
GNU LGPL (Lesser General Public Licence).
\end{thebibliography}
\end{document}